\documentclass[12pt]{article}
\usepackage{a4}
\usepackage{amsfonts,amssymb,amsmath}
\usepackage{graphicx}
\usepackage{cite}

\newtheorem{theorem}{Theorem}

\newtheorem{lemma}[theorem]{Lemma}

\newenvironment{proof}[1][Proof]{\textbf{#1.} }{\ \rule{0.5em}{0.5em}}
\begin{document}
\title{Nearly associative deformation quantization} \date{\empty}




\author{{Dmitri Vassilevich$^{1,2}$ and Fernando Martins Costa Oliveira$^{1}$}\\
{ $^1$CMCC, Universidade Federal do ABC,}\\ {Avenida dos Estados 5001, CEP 09210-580, Santo Andr\'e, SP, Brazil }\\
{$^2$Department of Physics, Tomsk State University, 634050 Tomsk, Russia }}

\maketitle


\begin{abstract}
We study several classes of non-associative algebras as possible candidates for deformation quantization in the direction of a Poisson bracket that does not satisfy Jacobi identities. We show that in fact alternative deformation quantization algebras require the Jacobi identities on the Poisson bracket and, under very general assumptions, are associative. At the same time, flexible deformation quantization algebras exist for any Poisson bracket.
\end{abstract}

\section{Introduction} \label{sec:intro}
The deformation quantization program considers a deformation of the algebra of functions on some manifold in the direction of a given Poisson bracket. The deformed algebra is an algebra of formal power series of the deformation parameter equipped with a new product, called the star product. In the classical setting \cite{BFFLS,DS,Waldmann}, the star product is assumed to be associative. Hence the Poisson bracket has to satisfy the Jacobi identities. 

However, some modern applications to magnetic backgrounds in field theory \cite{Jackiw:1984rd}, non-geometric backgrounds in string theory \cite{Blumenhagen:2010hj,Lust:2010iy,Gunaydin:2016axc} as well as some approaches to open strings and D-branes \cite{Cornalba:2001sm,Herbst:2003we} require quantization of Poisson brackets that do not satisfy the Jacobi identities, so that the corresponding algebras have to be non-associative.

Many examples of non-associative start products with varying degree of generality have been constructed \cite{Cornalba:2001sm,Herbst:2003we,Mylonas:2012pg,Bakas:2013jwa,Mylonas:2013jha,Kupriyanov:2015dda,Kupriyanov:2017oob}. However, it still remains unclear which condition can be used instead of the associativity in this context. The alternativity condition (total antisymmetry of the associator) may seem a natural choice that has been widely discussed in the literature, see e.g. \cite{Bojowald:2014oea,Kupriyanov:2016hsm,Kupriyanov:2017oob}. Despite many efforts not a single example of an alternative non-associative star product has been constructed. Moreover, according to \cite{Bojowald:2016lnl} the so-called monopole star products cannot be alternative.

In this paper we study the deformation quantization algebras that are called nearly associative in the monograph \cite{IvanBook}. These are alternative, right alternative and flexible algebras. They are defined by requesting that the associator is anti-symmetric in all arguments (alternative algebras) or anti-symmetric in some pairs of arguments (right alternative and flexible algebras). First, we analyze the Poisson bracket. It has been demonstrated in \cite{Kupriyanov:2016hsm} that for an alternative deformation quantization the Poisson bracket has to satisfy the Malcev identity. Here we prove a stronger statement that the Jacobi identity has to be satisfied in this case. Then, by using the Kleinfeld Theorem, we show that right alternative deformation quantization algebras are alternative. Our main result is that alternative star products are associative under very general assumptions. The situation with flexible algebras is very different: as we show for any Poisson bracket there is always a flexible star product.

\section{Main definitions} \label{sec:def}
Throughout this paper we shall deal with $C^\infty$ functions on $\mathbb{R}^n$ and thus with $C^\infty$ structures. 
A smooth bivector field $P$ defines a Poisson bracket\footnote{Summation over repeated indices is understood.} of smooth functions $f$ and $g$ as
\begin{equation}
\{ f,g\}=P(\mathrm{d}f,\mathrm{d}g)=P^{ij}\partial_i f \cdot \partial_j g \label{PB}
\end{equation}
which makes $C^\infty(\mathbb{R}^n)$ with the usual point-wise product $f\cdot g$ a Poisson algebra, $\{ f, g\cdot h\}=\{ f, g \} \cdot h + g\cdot \{f , g\}$. Let us define the Jacobiator as
\begin{equation}
\{f,g,h\}=\{ f,\{g ,h\}\}+\{ h,\{f ,g\}\}+\{ g,\{h ,f\}\}\,.\label{Jac}
\end{equation}
Any Poisson bracket is anti-symmetric, while any Jacobiator is totally anti-symmetric in all its' arguments.
If a Poisson bracket satisfies the Jacobi identity $\{f,g,h\}=0$, we call it a Poisson-Lie bracket. 

Frequently, the Poisson bracket for a generic bivector is called a quasi Poisson bracket, while the name Poisson bracket is used for Poisson-Lie brackets. We believe that our terminology (borrowed from the literature on nonassociaitve algebras) reflects better the underlying algebraic structure, especially since another type of the brackets is going to be introduced in the next section.

As in the classical paper \cite{BFFLS}, we define the star product as a product on the space of formal power series $C^\infty(\mathbb{R}^n)[[\lambda ]]$ by the formula
\begin{equation}
f\star g = f \cdot g +\sum_{r=1}^\infty \lambda^r C_r(f,g) \,,\label{starex}
\end{equation}
where $C_r$'s are bidifferential operators, and
\begin{equation}
C_1(f,g)=\{ f, g\} \,. \label{C1}
\end{equation}
Bidifferential operators are supposed to be differential, rather that pseudo-differential, operators in each of the arguments. Thus, the operators $C_r(f,\ . \ )$ and $C_r(\ . \ ,g)$ are local.
We shall denote the coefficients in front of $\lambda^p$ in the formal expansion of elements of $C^\infty(\mathbb{R}^n)[[\lambda ]]$ by subscripts in round brackets, $f=\sum_{r=o}^\infty \lambda^r f_{(r)}$. 

Equivalence of two star products $\star$ and $\star'$ means an algebra isomorphism that preserves the \emph{antisymmetrized} part of $C_1$. Since both star products should be expressed through bidifferential operators, we are left with the equivalence through Kontsevich gauge transformations \cite{Kontsevich:1997vb} $D(f\star' g)=(Df)\star (Dg)$ where $D=1+\sum_{r=1} \lambda^r D_r$ with some differential operators $D_r$.

One may add a symmetric part to $C_1(f,g)$, but such a part can be removed by passing to an equivalent product \cite{DS}. We also request that the unit function remains a unity of the deformed algebra. For associative deformations this is always true up to an equivalence \cite{DS}.

We shall call $C^\infty(\mathbb{R}^n)[[\lambda ]]$ endowed with a star product satisfying (\ref{starex}) and (\ref{C1})and with a unity given by the unit function a deformation quantization algebra (DQA).

Let us define an associator $A(f,g,h)$ through the equation
\begin{equation}
A(f,g,h)=f\star (g\star h)-(f\star g)\star h \,.\label{Ass}
\end{equation}

An DQA is called \emph{associative} if
\begin{equation}
A(f,g,h)\equiv 0 \,. \label{defass}
\end{equation}
There are three important classes of nearly associative algebras \cite{IvanBook}. 
A DQA is \emph{alternative} if $A(f,g,h)$ is totally anti-symmetric in all arguments, and it is \emph{flexible} if
\begin{equation}
A(f,g,f)\equiv 0\,.\label{difflex}
\end{equation}
A DQA is called \emph{right alternative} if it has the following identities
\begin{eqnarray}
&& A(f,g,g)=0\,, \label{ral}\\
&& ((f\star g)\star h)\star g=f \star ((g\star h)\star g) \,.\label{Mouf}
\end{eqnarray}
In a unital algebra, the second identity (\ref{Mouf}) (called the right Moufang identity) implies the first one (\ref{ral}).

\section{Poisson-Malcev bracket}\label{sec:PM}
It was demonstrated in \cite{Kupriyanov:2016hsm} that if a DQA is alternative, the corresponding Poisson bracket has to satisfy the Malcev identity
\begin{equation}
\{ h,f,\{h,g\}\}=\{ \{ h, f,g\},h\} \,,\label{MP}
\end{equation}
which makes it a Poisson-Malcev bracket. Here we prove the following

\begin{lemma}
Every Poisson-Malcev bracket is a Poisson-Lie bracket.
\end{lemma}
\begin{proof} For any Poisson-Malcev bracket one can prove the identity \cite[Corollary 1]{Sh2000}
\begin{equation}
\{ f,g,h \} \cdot \{ f, g \} =0 . \label{Ivan}
\end{equation}
By a partial linearization, $g\to g+d$, it can be transformed to another identity
\begin{equation}
\{ f,g,h \} \cdot \{ f, d \} + \{ f,d,h \} \cdot \{ f, g \} =0. \label{linI}
\end{equation}
In a local coordinate system
\begin{equation}
\{ f,g,h\} = J^{ijk}(x)\partial_if\, \partial_j g\, \partial_k h\,,\label{P3}
\end{equation}
where
\begin{equation}
J^{ijk}=P^{il}\partial_l P^{jk}+P^{jl}\partial_l P^{ki}+P^{kl}\partial_l P^{ij} . \label{JP}
\end{equation}
Let us suppose that the Jacobiator is not identically zero. Then, there is a point $x_0\in \mathbb{R}^n$ and there are three vectors $v_1$, $v_2$ and $v_3$ in the cotangent space $T^*_{x_0}$ such that 
\begin{equation}
J(x_0)(v_1,v_2,v_3)=J^{ijk}(x_0)v_{1,i}v_{2,j}v_{3,k}\neq 0 .\label{neq}
\end{equation}
Moreover, by (\ref{JP}), at least one of the vectors, say $v_1$, should have a non-zero contraction with $P(x_0)$. This means that there is another vector $v_4\in T^*_{x_0}$ such that
\begin{equation}
P(x_0)(v_1,v_4)\neq 0. \label{neq2}
\end{equation}
By taking $f=v_{1,j}(x^j-x_0^j)$, $g=v_{4,j}(x^j-x_0^j)$, $h=v_{3,j}(x^j-x_0^j)$ and $d=v_{2,j}(x^j-x_0^j)$ in some vicinity of $x_0$, we obtain from (\ref{Ivan}) and (\ref{linI}), respectively 
\begin{eqnarray}
&&P(x_0)(v_1,v_4)J(x_0)(v_1,v_4,v_3)=0,\nonumber\\
&&P(x_0)(v_1,v_4)J(x_0)(v_1,v_2,v_3)=-P(x_0)(v_1,v_2)J(x_0)(v_1,v_4,v_3).\nonumber
\end{eqnarray}
This equations clearly contradict (\ref{neq}) and (\ref{neq2}). Thus, $\{ f,g,h \}$ should vanish identically meaning that the bracket is a Poisson-Lie one.
\end{proof}

Some restrictions on Jacobiator following from the Malcev identity were obtained previously in \cite{Gunaydin:2013nqa}.

\section{Alternative and right alternative algebras}\label{sec:alt}
We start this section with a very simple but important statement.

\begin{lemma} A DQA does not have nilpotent elements. \label{L2}
\end{lemma}
\begin{proof}
Suppose that there is a non-zero element $f$ in $C^\infty (\mathbb{R}^n)[[\lambda ]]$ such that $f^k_\star =0$, $k\in \mathbb{N}$. (The power is calculated with the star product). Let $\lambda^r f_{(r)}$ be the lowest non vanishing order in the $\lambda$-expansion. Then, $f_\star^k=\lambda^{rk} f^k_{(r)} +\mathcal{O}(\lambda^{rk+1})$, where $f^k_{(r)}$ is computed with the point wise product. Since $f^k_{(r)}=0$ implies $f_{(r)}=0$, we have a contradiction.
\end{proof}

The case of right alternative algebras may be treated with the help of the following Theorem \cite{IvanBook}.

\begin{theorem}[Kleinfeld] Every right alternative algebra without nilpotent elements is alternative.
\end{theorem}

Together with the Lemma \ref{L2} this Theorem implies that any right alternative DQA is alternative.

We are going to demonstrate in this section that alternative DQAs are almost always associative. To this end, let us study restrictions on the associator. Alternative algebras have a number of identities. One of them reads \cite[p.\ 149]{IvanBook}
\begin{equation}
\left( A([g,h]^2_\star,r,s) \right)^2_\star =0\,, \label{id}
\end{equation}
where 
\begin{equation}
[g,h]:=g\star h - h\star g .\label{commut}
\end{equation} 
By Lemma \ref{L2}, the identity (\ref{id}) yields another identity
\begin{equation}
A([g,h]^2_\star,r,s)  =0\,. \label{id2}
\end{equation}
Due to the results of previous section, we may use the heavy machinery of Poisson geometry\footnote{We should again warn the reader on some terminology mismatch between the present paper and some other literature, e.g. \cite{Vaisman,Waldmann}.} \cite{Vaisman}. In particular, the Wienstein splitting Theorem \cite{We1983,Vaisman,Waldmann} affirms (after being translated to our terminology) that if $\{ . , . \}$ is a Poisson-Lie bracket, each point $x_0$ has an open neighborhood $U$ with a centered coordinate chart $(q^a,p^a,y^\alpha)$ such that 
\begin{equation}
\{ f, g\}=\sum_a \bigl( \partial_{q^a} f \partial_{p^a} g - \partial_{q^a} g \partial_{p^a} f\bigr) + \sum_{\alpha,\beta} \varphi^{\alpha\beta}(y) \partial_{y^\alpha}f \partial_{y^\beta}g \label{split}
\end{equation} 
and $\varphi^{\alpha\beta}$ vanishes at $x_0$. It is important, that if the Poisson-Lie bivector does not vanish at $x_0$, there is a least one pair of the coordinates of the type $q^a,p^a$, call it $(q^1,p^1)$. For any given smooth function $\psi$ there is always a function $g\in C^\infty(\mathbb{R}^n)$ such that the  differential equation 
\begin{equation*}
\partial_{q^1} g = \tfrac 12 \psi 
\end{equation*}
holds in $U$. Thus, by taking $h=p^1$ in $U$, we conclude that for any smooth $\psi$ there is a pair $g,h$ such that on $U$
\begin{equation}
\{ g,h\} = \tfrac 12 \psi \qquad \mbox{and} \qquad [g,h]=\lambda \psi +\mathcal{O}(\lambda^2)\,. \label{gh}
\end{equation}

Let us now show that any element of the algebra can be approximated near $x_0$ by a sum of squared commutators and constants.
Take arbitrary $f\in C^\infty(\mathbb{R}^n)[[\lambda]]$ and consider its' formal expansion $f=f_{(0)}+\sum_{l=1}^\infty \lambda^l f_{(l)}$. Let us fix some positive integer $k$ and request that
\begin{equation}
f=\sum_{l=0}^k \bigl( \lambda^l c_l + \lambda^{l-2} [g_{l},h_{l}]^2_\star \bigr) +\mathcal{O}(\lambda^{k+1}) \label{approx}
\end{equation}
in $U$ for some constants $c_1$ and functions $g_l$, $h_l$. From now on we need $U$ being bounded. If it is not, we may pass to a bounded sub-neighborhood of $x_0$ without affecting any of the statements made above. At the zeroth order of $\lambda$ Eq.\ (\ref{approx}) reads
\begin{equation}
f_{(0)}=c_0 + ([g_0,h_0]_{(1)})^2 \,. \label{0th}
\end{equation}
$f_{(0)}$ is a smooth function $\mathbb{R}^n$. Thus it is bounded on (a bounded domain) $U$. The constant $c_0$ can be used to shift $f_{(0)}$ away from $0$ on $U$, so that the square root
\begin{equation}
\sqrt{ f_{(0)} -c_0} = [g_0,h_0]_{(1)} \label{sol0}
\end{equation}
is smooth. Thus, by (\ref{gh}), there is a pair of functions $g_0,h_0$ such that eq.\ (\ref{0th}) is satisfied, and (\ref{approx}) is true to the order $\lambda^0$. In a similar manner, one shows that the order $j$ of Eq.\ (\ref{approx})
\begin{equation}
f_{(j)}=c_j + ([g_j,h_j]_{(1)})^2 +\sum_{m=0}^{j-1} \bigl( [g_m,h_m]^2_\star)_{(j-m+2)} \label{jth}
\end{equation}
can be satisfied by a suitable choice of $c_j$, $g_j$ and $h_j$.

By substituting (\ref{approx}) in the identity (\ref{id2}) one shows, that on $U$ the associator of three arbitrary functions $f$, $r$ and $s$ vanishes at least to the order $k$:
\begin{equation}
A(f,r,s)_{(j)}=0\qquad \mbox{for} \qquad j\leq k. \label{idk}
\end{equation}
Since $k$ is arbitrary, one can take the limit $k\to \infty$ (which is an honest limiting procedure in the $\lambda$-adic topology of $C^\infty (\mathbb{R}^n)[[\lambda]]$) to obtain

\begin{lemma}\label{L4} Let $P$ be a Poisson-Lie bivector which does not vanish at some point $x_0\in \mathbb{R}^n$. Then there is an open neighborhood $U$ of $x_0$ such that for arbitrary alternative DQA that quantizes $P$ the associator vanishes in $U$.
\end{lemma}

There two possible ways to make the statement of Lemma \ref{L4} global. First, consider the case when $P$ is non-vanishing almost everywhere (on a dense subset of $\mathbb{R}^n$). Since $A(f,r,s)$ is a continuous function for any smooth $f$, $r$ and $s$, one can extend the identity $A(f,r,s)=0$ to the whole $\mathbb{R}^n$ by continuity. 

In the case which is opposite to the one mentioned in the previous paragraph, $P$ vanishes identically in some ball in $\mathbb{R}^n$. However, if we assume that the symbols of bidifferential operators $C_r$ are local polynomials of $P$ and its' derivatives, all $C_r(f,g)$ vanish in the ball. Thus, only the point wise term remains in the star product (\ref{starex}), so that the associator vanishes as well. 

We arrive at the main result of this paper.

\begin{theorem}\label{T5}
Suppose that an alternative DQA satisfies at least one of the following conditions.\\
(a) The bivector $P$ is non-vanishing on a dense set in $\mathbb{R}^n$.\\
(b) The symbols of bidifferential operators $C_r$ are local polynomials of $P$ and its' derivatives.\\
Then such DQA is associative.
\end{theorem}

Note that if an alternative DQA does not satisfy the condition (b), but an equivalent algebra does, both algebras are associative. Thus, the family of star products to which the condition (b) applies is indeed very large. For example, by the Kontsevich Theorem \cite{Kontsevich:1997vb} all associative star products are in this family. 

Theorem \ref{T5} implies Theorem 1 of \cite{Bojowald:2016lnl} as a particular case.

\section{Flexible algebras}\label{sec:flex}
To discuss the case of flexible algebras, we need a characteristic property of these algebras \cite{An1966}, see also \cite{BO1981}. Let $B$ be some algebra with the product $x,y\mapsto xy$. Define another algebra $B^+$ on the same linear space with a symmetrized (Jordan) product $x\circ y =xy+yx$. The algebra $B$ is flexible if and only if the commutator $[x,y]=xy-yx$ is a derivation of $B^+$. As has been noticed in \cite{Bojowald:2016lnl}, precisely this property makes flexible algebras physically relevant: one can consistently define the evolution equations on $B^+$ with the help of commutator in $B$.

In contrast to the case of alternative star products, flexible but non-associative star products do always exist. For example, the product
\begin{equation}
f\star g = f\cdot g +\lambda \{ f,g\}. \label{flexstar}
\end{equation}
with $C_r \equiv 0$ for $r\geq 2$ is flexible for any bivector $P$ and is not associative even if $P$ is Poisson-Lie.

We are not aware of any other examples of flexible DQAs.

\section{Conclusions}
In this paper, we analyzed nearly associative algebras as candidates for DQAs. We have demonstrated that alternative DQAs require Jacobi identity on the Poisson structure, which makes them useless in the context of non-geometric fluxes in string theory. Moreover, under very general assumptions alternative and right alternative DQAs are associative. From a practical point of view, this means the end of the story of alternative deformation quantization. On the contrary, a flexible deformation quantization exists for any Poisson bracket.

\section*{Acknowledgments}
We are grateful to Ivan Shestakov, Vlad Kupriyanov and Richard Szabo for fruitful discussions. One of the authors (D.V.) was supported in parts by by the grant 2016/03319-6 of the S\~ao Paulo Research Foundation (FAPESP),  by the grant 303807/2016-4 of CNPq, by the RFBR project 18-02-00149-a, and by the Tomsk State University Competitiveness Improvement Program. The other author (F.M.C.O.) was supported by CAPES.


\begin{thebibliography}{99}
\bibitem{An1966}
Anderson, T.: A note on derivations of commutative algebras.
Proc.\ Amer.\ Math.\ Soc. {\bf 17}, 1199-1202 (1966).

\bibitem{Bakas:2013jwa} 
  Bakas, I., L\"ust, D.:
  3-Cocycles, Non-Associative Star-Products and the Magnetic Paradigm of R-Flux String Vacua.
  JHEP {\bf 1401}, 171 (2014)
  doi:10.1007/JHEP01(2014)171
  [arXiv:1309.3172 [hep-th]].

\bibitem{BFFLS}
Bayen, F., Flato, M., Fronsdal, C., Lichnerowicz, A., Sternheimer, D.:
  Deformation Theory and Quantization. 1. Deformations of Symplectic Structures.
  Annals Phys.\  {\bf 111},  61 (1978).
	
\bibitem{BO1981}
Benkart, G. M., Osborn, J. M.: Flexible Lie-admissible algebras.
J.\ Algebra {\bf 71}, 11-31 (1981).

\bibitem{Blumenhagen:2010hj} 
  Blumenhagen, R., Plauschinn, E.:
  Nonassociative Gravity in String Theory?
  J.\ Phys.\ A {\bf 44}, 015401 (2011)
  doi:10.1088/1751-8113/44/1/015401
  [arXiv:1010.1263 [hep-th]].

\bibitem{Bojowald:2014oea} 
  Bojowald, M., Brahma, S., Buyukcam U., Strobl,T.:
  States in non-associative quantum mechanics: Uncertainty relations and semiclassical evolution.
  JHEP {\bf 1503}, 093 (2015)
  doi:10.1007/JHEP03(2015)093
  [arXiv:1411.3710 [hep-th]].

\bibitem{Bojowald:2016lnl}
  Bojowald, M., Brahma, S., Buyukcam U., Strobl,T.:
  Monopole star products are non-alternative.
  JHEP {\bf 1704} 028 (2017) 
  doi:10.1007/JHEP04(2017)028
  [arXiv:1610.08359 [math-ph]].

\bibitem{Cornalba:2001sm} 
  Cornalba, L., Schiappa, R.:
  Nonassociative star product deformations for D-brane world volumes in curved backgrounds,
  Commun.\ Math.\ Phys.\  {\bf 225}, 33 (2002)
  doi:10.1007/s002201000569
  [hep-th/0101219].
	
\bibitem{DS} Dito, G., Sternheimer, D.: Deformation quantization:
Genesis, developments and metamorphoses.  In: IRMA Lect. Math. Theor.
Phys., 1 pp 9--54. Berlin, de Gruyter, (2002). arXiv:math/0201168.

\bibitem{Gunaydin:2016axc} 
  Gunaydin, M, Lust, D., Malek, E.:
  Non-associativity in non-geometric string and M-theory backgrounds, the algebra of octonions, and missing momentum modes,
  JHEP {\bf 1611}, 027 (2016)
  doi:10.1007/JHEP11(2016)027
  [arXiv:1607.06474 [hep-th]].

\bibitem{Gunaydin:2013nqa}
  G\"unaydin, M., Minic, D.:
  Nonassociativity, Malcev Algebras and String Theory.
  Fortsch.\ Phys.\  {\bf 61} (2013) 873
  doi:10.1002/prop.201300010
  [arXiv:1304.0410 [hep-th]].

\bibitem{Herbst:2003we} 
  Herbst, M., Kling, A., Kreuzer, M.:
  Cyclicity of nonassociative products on D-branes.
  JHEP {\bf 0403}, 003 (2004)
  doi:10.1088/1126-6708/2004/03/003
  [hep-th/0312043].

\bibitem{Jackiw:1984rd} 
  Jackiw, R.:
  3 - Cocycle in Mathematics and Physics.
  Phys.\ Rev.\ Lett.\  {\bf 54}, 159 (1985).
  doi:10.1103/PhysRevLett.54.159
	

\bibitem{Kontsevich:1997vb} 
  Kontsevich, M.:
  Deformation quantization of Poisson manifolds. 1.
  Lett.\ Math.\ Phys.\  {\bf 66}, 157 (2003)
  doi:10.1023/B:MATH.0000027508.00421.bf
  [q-alg/9709040].

\bibitem{Kupriyanov:2016hsm}
 Kupriyanov, V.~G.:
  Weak associativity and deformation quantization,
  Nucl.\ Phys.\ B {\bf 910},  240 (2016).
  doi:10.1016/j.nuclphysb.2016.07.004
  [arXiv:1606.01409 [hep-th]].

\bibitem{Kupriyanov:2017oob} 
  Kupriyanov, V.~G., Szabo, R.~J.:
  G$_{2}$-structures and quantization of non-geometric M-theory backgrounds,
  JHEP {\bf 1702}, 099 (2017)
  doi:10.1007/JHEP02(2017)099
  [arXiv:1701.02574 [hep-th]].
	
\bibitem{Kupriyanov:2015dda} 
  Kupriyanov, V.~G., Vassilevich, D.~V.:
  Nonassociative Weyl star products.
  JHEP {\bf 1509}, 103 (2015)
  doi:10.1007/JHEP09(2015)103
  [arXiv:1506.02329 [hep-th]].
	
\bibitem{Lust:2010iy} 
  L\"ust, D.:
  T-duality and closed string non-commutative (doubled) geometry.
  JHEP {\bf 1012}, 084 (2010)
  doi:10.1007/JHEP12(2010)084
  [arXiv:1010.1361 [hep-th]].

\bibitem{Mylonas:2012pg} 
  Mylonas, D., Schupp, P., Szabo, R.~J.:
  Membrane Sigma-Models and Quantization of Non-Geometric Flux Backgrounds.
  JHEP {\bf 1209}, 012 (2012)
  doi:10.1007/JHEP09(2012)012
  [arXiv:1207.0926 [hep-th]].

\bibitem{Mylonas:2013jha} 
  Mylonas, D., Schupp, P., Szabo, R.~J.:
  Non-Geometric Fluxes, Quasi-Hopf Twist Deformations and Nonassociative Quantum Mechanics.
  J.\ Math.\ Phys.\  {\bf 55}, 122301 (2014)
  doi:10.1063/1.4902378
  [arXiv:1312.1621 [hep-th]].

\bibitem{Sh2000}
Shestakov, I.~P.:
Speciality problem for Malcev algebras and Poisson Malcev algebras. In: Costa R.\ et al (eds.)
IV Conference on non-associative algebra and its applications, 2000, S\~ao Paulo,
Marcel Dekker, NY (2000).

\bibitem{Vaisman}
Vaisman, I.: Lectures on the geometry of Poisson manifolds. Birkh\"auser, Basel (1994).

\bibitem{Waldmann}
Waldmann, S.: Poisson-Geometrie und Deformationsquantisierung. Springer, Berlin (2007).

\bibitem{We1983}
Weinstein, A.: The local structure of Poisson manifolds. 
J.\ Diff.\ Geom.\ {\bf 18}, 523-557 (1983).

\bibitem{IvanBook}
Zhevlakov, K.~A., Slin'ko, A.~M., Shestakov, I.~P., Shirshov, A.~I.:
Rings that are nearly associative. Academic Press, New York (1982).
\end{thebibliography}
\end{document}